\documentclass[conference]{IEEEtran}
\IEEEoverridecommandlockouts
\usepackage{cite}
\usepackage{amsmath,amssymb,amsfonts}
\usepackage{amsthm}
\usepackage{float}
\usepackage{comment}
\usepackage{algorithmic}
\usepackage{graphicx}
\usepackage{textcomp}
\usepackage{xcolor}
\usepackage{balance}
\usepackage{hyperref}
\usepackage{subcaption}
\def\BibTeX{{\rm B\kern-.05em{\sc i\kern-.025em b}\kern-.08em
    T\kern-.1667em\lower.7ex\hbox{E}\kern-.125emX}}

\newtheorem{theorem}{Theorem}

\theoremstyle{definition}

\newtheorem{lemma}[theorem]{Lemma}

\theoremstyle{definition}
\newtheorem{exmp}{Example}[section]    
    
\begin{document}

\title{TreeExplorer: a coding algorithm for rooted trees with application to wireless and ad hoc routing}

\author{\IEEEauthorblockN{Amirmohammad Farzaneh\IEEEauthorrefmark{1},
Mihai-Alin Badiu\IEEEauthorrefmark{2}, and Justin P. Coon\IEEEauthorrefmark{3}}
\IEEEauthorblockA{Department of Engineering Science,
University of Oxford\\
Oxford, UK\\
Email: \IEEEauthorrefmark{1}amirmohammad.farzaneh@eng.ox.ac.uk,
\IEEEauthorrefmark{2}mihai.badiu@eng.ox.ac.uk,
\IEEEauthorrefmark{3}justin.coon@eng.ox.ac.uk}}

\maketitle

\begin{abstract}
Routing tables in ad hoc and wireless routing protocols can be represented using rooted trees. The constant need for communication and storage of these trees in routing protocols demands an efficient rooted tree coding algorithm. This efficiency is defined in terms of the average code length, and the optimality of the algorithm is measured by comparing the average code length with the entropy of the source. In this work, TreeExplorer is introduced as an easy-to-implement and nearly optimal algorithm for coding rooted tree structures. This method utilizes the number of leaves of the tree as an indicator for choosing the best method of coding.  We show how TreeExplorer can improve existing routing protocols for ad hoc and wireless systems, which normally entails a significant communication overhead.
\end{abstract}


\section{Introduction}

Trees can be seen in many areas of science and technology. Some of the most well-known application areas of trees are phylogenetic trees \cite{felsenstein2004inferring}, XML documents \cite{dykes2011xml}, and binary search trees \cite[Ch.~12]{cormen2022introduction} . They can also be seen in network routing protocols. In path-vector routing protocol \cite{medhi2017network}, which is used in the Border Gateway Protocol \cite{sobrinho2003network}, each node maintains a table that contains all the hops in the shortest path from that node to all other nodes in the network. As a result, it can be shown that path-vector routing tables are essentially rooted trees.

Routing is of great importance in ad hoc and wireless networks. Some of the existing methods for routing in ad hoc networks explicitly make use of the path-vector routing protocol \cite{chau2008inter, rangarajan2004using}. Other methods that do not specify the use of such routing tables can also benefit from replacing link state or distance vector routing tables with path-vector routing tables to avoid looping \cite{zhang2007ad}. This observation conveys the inherent existence and use of rooted trees in ad hoc and wireless routing protocols. Additionally, some ad hoc routing protocols are entirely based on rooted trees in order to be scalable \cite{ioannidis2004scalable, ioannidis2005high}. Therefore, it can be seen that rooted trees play a significant role in ad hoc and wireless routing protocols.

Trees (and graphs in general) are known to be complex data structures \cite{farzaneh2021kolmogorov}. This complexity might not be felt in networks with small number of nodes. However, networks are dynamic entities \cite{farzaneh2022information}, and we are rapidly moving towards extremely big networks. With the emergence of IoT and 6G, the number of nodes is soon expected to be around 10 million devices per $\text{km}^2$ \cite{nakamura20205g}. This fast growth will lead to more complex routing tables, which means significantly bigger rooted trees corresponding to each routing table.

In all of the application areas that were mentioned above for trees, they are stored in a database, communicated through a channel, or both. This would necessitate the need of coding the trees at hand for storage and communication purposes. It is known that routing protocols in ad hoc and wireless systems add a  significant overhead to the communications in the system \cite[Ch.~16]{goldsmith2005wireless}, and this overhead will only grow larger with the exponential growth of the size of networks over time. As routing tables can be represented using rooted trees, we look for coding algorithms for rooted trees for use in routing protocols. If the trees are coded in the most compact way possible, the communication overhead in these systems will be significantly reduced. Consequently, the aim of this paper is to provide an efficient coding algorithm for rooted trees. Additionally, we will study the underlying structure of the tree independently of the node labels. It is shown that this will improve the performance of the algorithm.

In this paper, we introduce Pit-climbing, Tunnel-digging, and TreeExplorer as novel coding algorithms for rooted trees. As a hybrid of Pit-climbing and Tunnel-digging, TreeExplorer will be our ultimate proposed method. The performance of the algorithm is evaluated using information theoretic metrics, and simulation results illustrate its near-optimal performance. Ultimately, we discuss how TreeExplorer can be used to reduce the communication overhead in ad hoc and wireless routing.

\section{Tree sources}
\label{sources}

In this paper, we consider trees to be rooted and unlabeled. We use the terms `unlabelled tree' and `tree structure' interchangeably. We refer to movements along the edges towards the root as \textit{upwards}, and \textit{downwards} otherwise. Every node, except for the root, has a unique \textit{parent}, which is the node directly connected to it in the upward direction. We use the term \textit{siblings} to refer to a group of nodes that share the same parent. By \textit{depth} of a node we mean the length of the shortest path from that node to the root. The children of each node are unordered unless explicitly mentioned. For unordered trees, we always arrange the children of each node from left to right using a mapping from unlabeled unordered rooted trees to unlabeled ordered rooted trees. The term \textit{entropy} is used in this paper to refer to Shannon's entropy \cite{shannon1948mathematical}. All the entropies in this paper are calculated in base two (bits).

We are interested in calculating the entropy of a uniform source for rooted tree structures. Only uniform sources are considered as the authors were not able to find random generation models designed solely for tree structures. The sequence of the number of unlabeled unordered rooted trees with $n$ nodes is listed on the on-line encyclopedia of integer sequences \cite[A000081]{oeis}. Table \ref{sequence} shows the possible trees of this kind that can be built with up to three nodes. It is known that the asymptotic limit of this sequence is $cd^nn^{-3/2}$ \cite{polya2012combinatorial}, where $c$ and $d$ are constants that can be found at \cite[A187770]{oeis} and \cite[A051491]{oeis}, respectively. Consequently, the asymptotic uniform entropy for rooted tree structures with $n$ nodes can be calculated using the following equation.

\begin{align}
\begin{split}
\label{asymp_entropy}
    H(G)&\sim n\log_2d-1.5\log_2n+\log_2c\\
    &\approx 1.5635n-1.5\log_2n-1.1846
\end{split}
\end{align}

Eq. (\ref{asymp_entropy}) shows that the growth rate of the entropy of a uniformly distributed unlabeled unordered rooted tree source is asymptotically linear. This statement also holds for unlabeled ordered rooted trees, as their number matches the sequence of Catalan numbers \cite[Ch.~8]{koshy2008catalan}. Additionally, the number of labeled (unrooted) trees is shown to be $n^{n-2}$ \cite[p.~26]{cayley_2009}, which implies that $\sim n\log_2n$ bits are needed to code random labeled trees with uniform distribution.

\begin{table}
\centering
\caption{Possible unlabeled unordered rooted trees of up to three nodes}
\label{sequence}
\begin{tabular}{ |c| c| c| }
\hline
 $\boldsymbol{n}$ & \textbf{Possible trees} & \textbf{Count}\\ 
 \hline
 1 & \begin{minipage}{.4\columnwidth}
      \includegraphics[width=\linewidth]{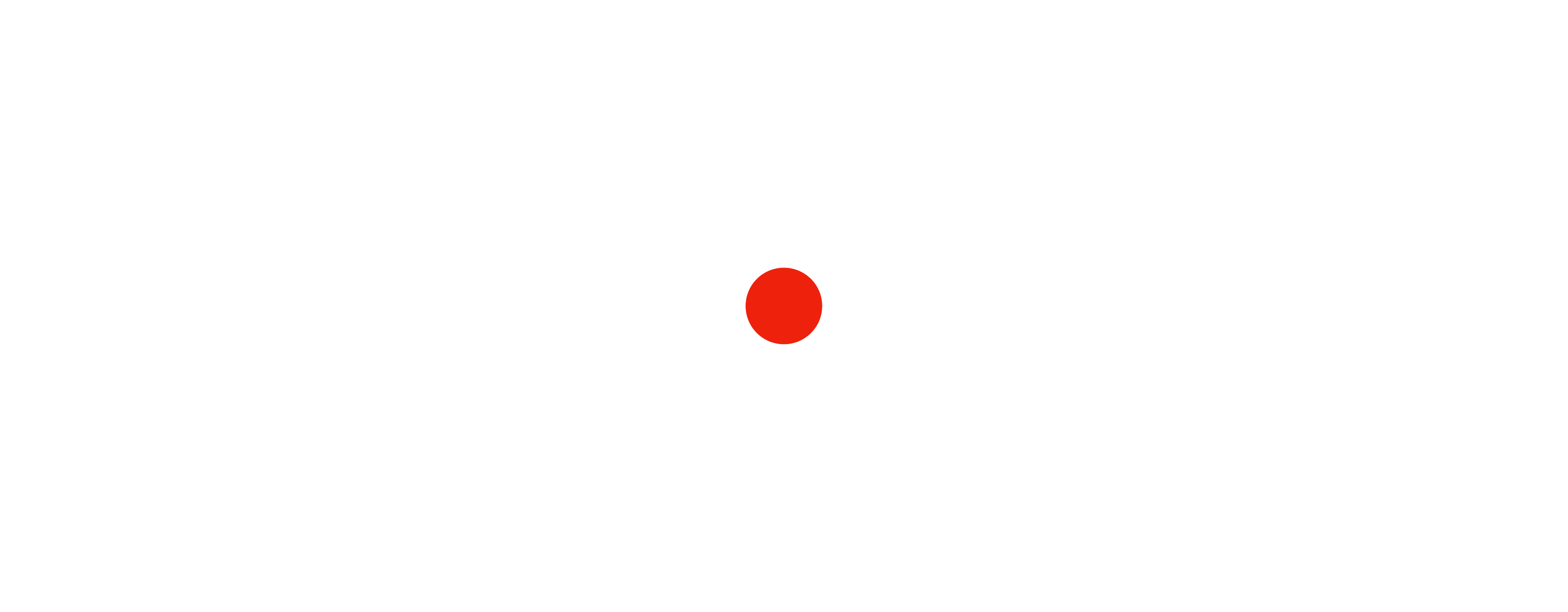}
    \end{minipage}
& 1\\
\hline
2 &  \begin{minipage}{.4\columnwidth}
      \includegraphics[width=\linewidth]{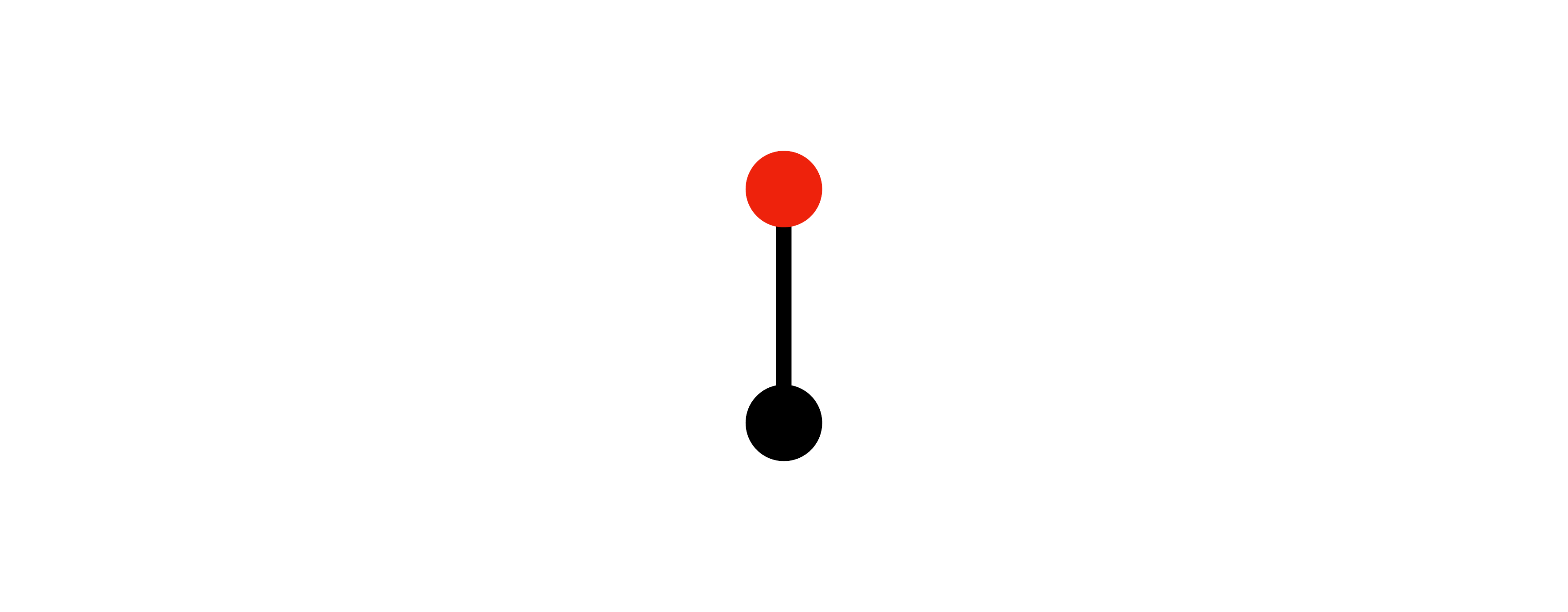}
    \end{minipage} & 1\\
    \hline
3 &  \begin{minipage}{.4\columnwidth}
      \includegraphics[width=\linewidth]{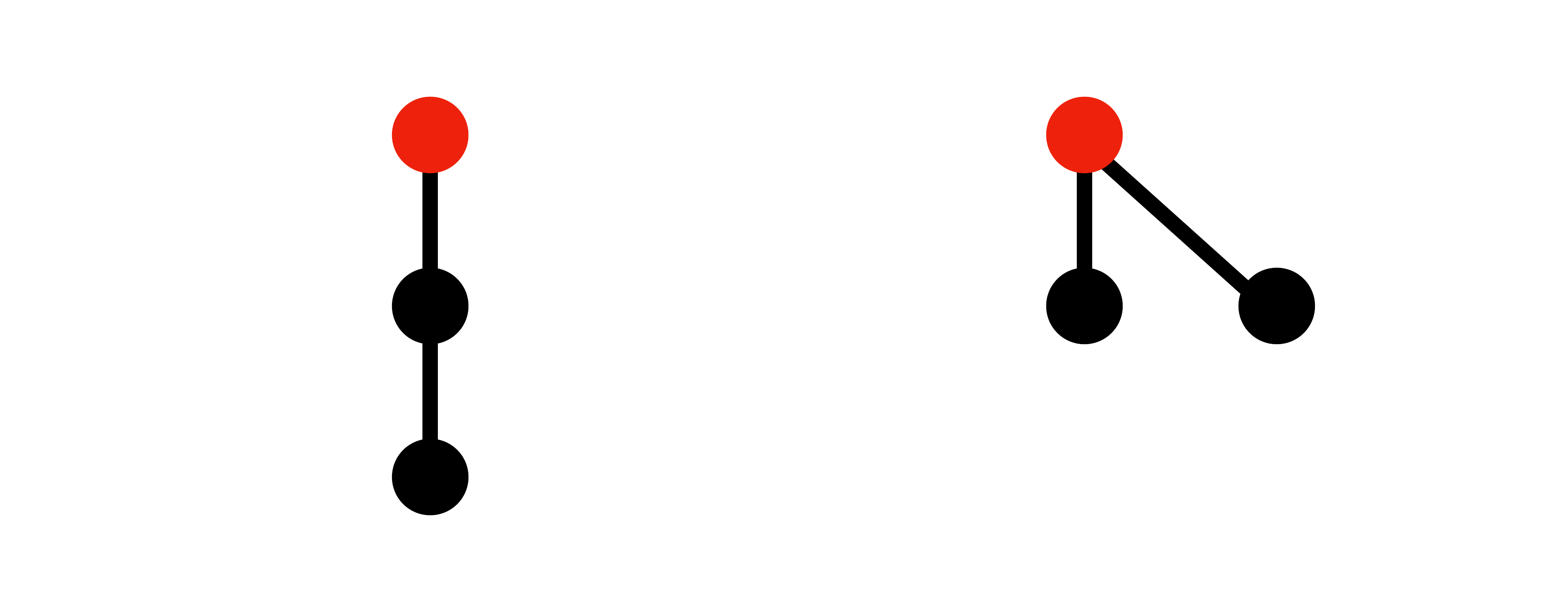}
    \end{minipage} & 2\\
    \hline
\end{tabular}
\end{table}

\section{Pit-climbing algorithm}
\label{pit-climbing}
In this section, we introduce a novel tree structure coding algorithm that we call \textit{pit-climbing}. We use this term because of the analogy between the proposed method, and a climber that has been trapped in a pit and wants to climb up.

\textbf{Ternary pit-climbing algorithm (TPC): } We start traversing the tree from the leftmost leaf. We log our tree traversal using three symbols: $\uparrow$, $\Uparrow$ and $\downarrow$. Anytime that we are at a leaf, we take the only possible path, which is upwards. If we take an upward path at any point from an edge, we consider that edge and the subtree below it as deleted (or filled-in) from the original tree so that we do not explore it again. Additionally, we log this upward movement in our code. If we have moved to a node that we have never been to before, we log a $\uparrow$ in the code. Otherwise, if we move upwards to a node that we have seen before, we log it with a $\Uparrow$. When we reach a node that is not a leaf, we look at the leaves of the rooted subtree whose root is the node we are currently at. We then take the path downwards that falls into the leftmost leaf of that subtree. We log this entire fall with a single $\downarrow$. We continue exploring the tree and logging the code in the same manner until we reach the root of the tree and there is no other edge to fall into.

We will clarify TPC with the following example.

\begin{exmp}
    Assume that we are given the rooted tree structure of Fig. \ref{tpc}, where the red node shows the root. The starting point of the algorithm is indicated, and the arrows show the path that PC takes. The orange, green, and blue arrows are used to show $\downarrow$, $\uparrow$, and $\Uparrow$, respectively.
    \begin{figure}[H]
        \centering
        \includegraphics[width=\columnwidth]{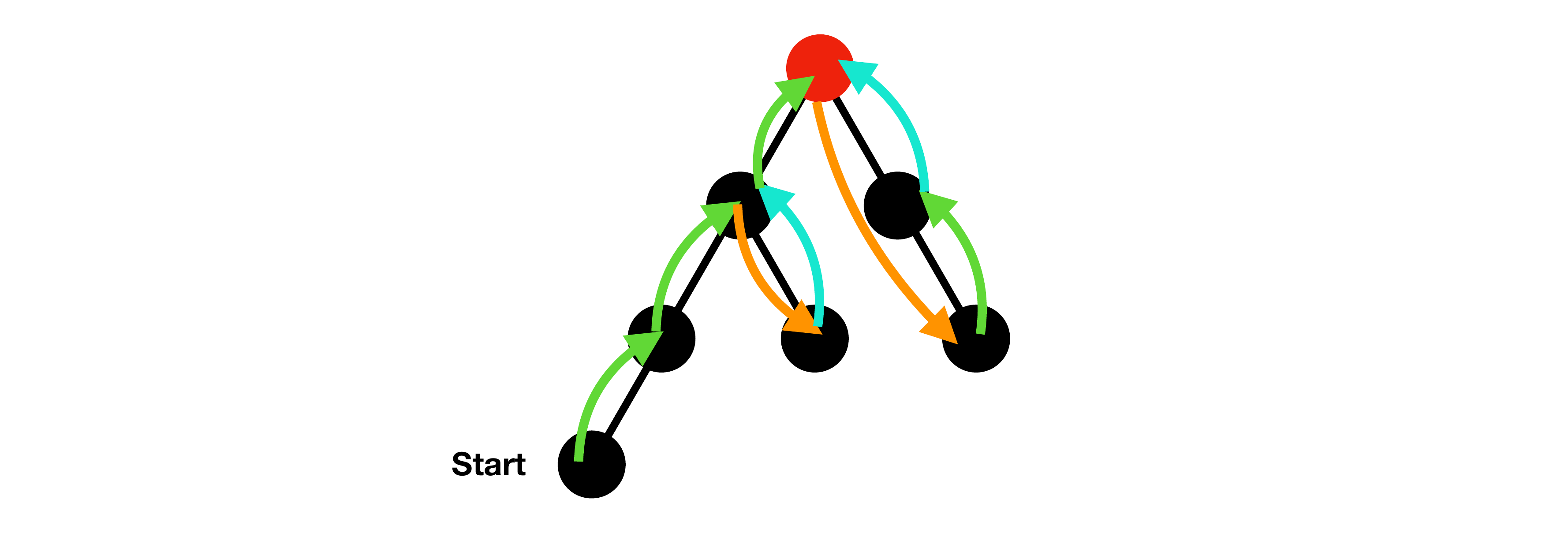}
        \caption{Running TPC on a sample tree}
        \label{tpc}
    \end{figure}
    
\end{exmp}

In source coding, usually a binary code is preferred over a ternary code, as most of our systems for storage and communication are binary-based. To transform TPC codes into binary, we look back at the definition of the symbols. We observe that we can never have consecutive $\downarrow$s. This is because whenever we fall, we fall down to a leaf, so we can never fall twice or more. We make use of this fact, and assign 0 to $\downarrow$, and 00 to $\Uparrow$. We also use 1 to represent $\uparrow$. We call this new binary code for rooted tree structures simply \textit{pit-climbing} (PC). Even though this method of coding does not provide us with an instantaneous code, we claim that PC creates uniquely decodable codes. Theorem \ref{uniquely} proves this statement.

\begin{theorem}
\label{uniquely}
Codes generated by pit-climbing are uniquely decodable.
\end{theorem}
\begin{proof}
We use induction on the depth of the tree. Firstly, the code for a tree with a single node is uniquely decodable ($\emptyset$). Next, assume that we know PC codes for all trees with a depth of $k$ or less are uniquely decodable. For a tree with a depth of $k+1$, we look at the subtrees of the children of the root. Based on the induction, we know that the PC code for all these subtrees are uniquely decodable. The PC code for the original tree is the concatenation of the PC codes of the subtrees, with a connector of $\uparrow\downarrow \equiv 10$ for the first two subtrees, and $\Uparrow\downarrow \equiv 000$ for all other subtrees. In case the root only has one child, the final code will be the code of the subtree rooted at the child, plus an additional $\uparrow \equiv 1$. Therefore, in all of the cases, the PC code of the tree with a depth of $k+1$ can be uniquely decoded.
\end{proof}

Furthermore, we would like to investigate the length of the codewords generated by PC.

\begin{theorem}
\label{PClen}
The PC codeword length for a rooted tree structure with $n$ nodes and $l$ leaves is $n+2l-3$ bits.
\end{theorem}
\begin{proof}
Firstly, notice that each $\downarrow$ in the TPC code corresponds to a leaf, as we always fall into a leaf. Additionally, we fall into every leaf except for the one we start the algorithm from exactly once. Therefore, we have $l-1$ $\downarrow$s in the TPC code, which translates into $l-1$ bits in the PC code. Additionally, we climb up each edge of the tree exactly once. Therefore, the number of $\uparrow$s and $\Uparrow$s in the TPC code is equal to the number of edges, which is $n-1$. However, every $\Uparrow$ will translate into 2 bits in the PC code. The number of $\Uparrow$s is equal to the number of $\downarrow$s, as anytime we fall from a node we will have to climb back up to it at some point. Therefore, the number of $\Uparrow$s is also $l-1$, and we will have $l-1$ additional bits when translating the TPC code into a PC code. Consequently, the total number of bits in the PC code will be $l-1+n-1+l-1=n+2l-3$.
\end{proof}

\section{Tunnel-digging algorithm}

Based on Theorem \ref{PClen}, the PC code length will increase with the number of leaves. However, the number of rooted trees with $l$ leaves does not necessarily increase with $l$. Hence, there is no justified reason for having longer codes for trees with more leaves. As a result, another algorithm called \textit{tunnel-digging} is developed to tackle this problem. The PC algorithm is based on traversing a tree along its edges, up and down. Whereas, in TD, the aim is to traverse the tree in a horizontal manner. The name \textit{tunnel-digging} comes from seeing the traversal method of this algorithm as digging tunnels between nodes on the same depth.

\textbf{Ternary tunnel-digging algorithm (TTD):} We start with the leftmost child of the root, and start moving right to nodes with the same depth in the order of the nodes. For each node that we encounter, we log a $\leftarrow$ if it is a leaf, and a $\rightarrow$ otherwise. If at any point we have to move between two nodes that are not siblings, we use a $\Rightarrow$ to show the transition (digging a tunnel!). Additionally, if at any point there are no more nodes on the right to move to, we move to the leftmost node on the level below, and we mark this transition again with a $\Rightarrow$. We continue until all the leaves of the tree are logged in the code.

The following example illustrates running TTD on a sample rooted tree structure.

\begin{exmp}
Assume that we are given the rooted tree structure of Fig. \ref{ttd}, where the red node shows the root. The starting point of the algorithm is indicated, and the arrows show the path that TTD takes. The blue, green, and orange arrows are used to show $\leftarrow$, $\rightarrow$, and $\Rightarrow$, respectively.

\begin{figure}[H]
        \centering
        \includegraphics[width=\columnwidth]{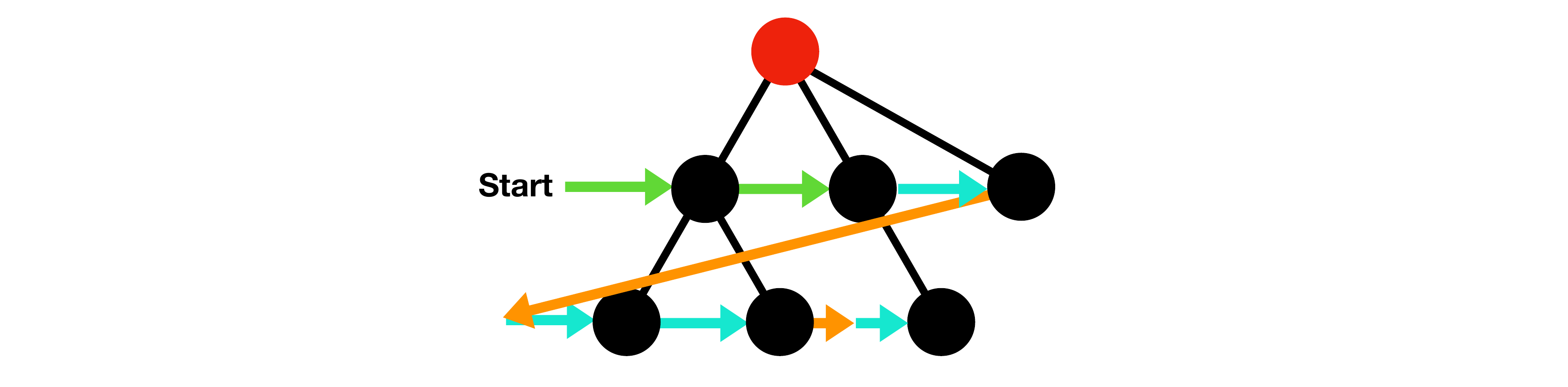}
        \caption{Running TTD on a sample tree}
        \label{ttd}
\end{figure}

\end{exmp}

To transform the TTD code into a binary code which will be called tunnel-digging (TD), we again use the properties of the TTD symbols. We notice that we can never have two consecutive $\Rightarrow$s, as there is always at least one node in between the dug tunnels. Therefore, we use $0$ to represent $\Rightarrow$ and $00$ to represent $\rightarrow$. Additionally, we use $1$ to show $\leftarrow$. Notice that we use a shorter code for leaf nodes, as the number of leaf nodes is expected to be higher than the number of non-leaf nodes when TD is used to code the tree. The proof that the code is uniquely decodable can be done in the same manner as Theorem \ref{uniquely} by replacing $\downarrow$, $\uparrow$, and $\Uparrow$ with $\rightarrow$, $\leftarrow$, and $\Rightarrow$, respectively. The following Theorem calculates the code length of tunnel-digging.

\begin{theorem}
The TD codeword length for a rooted tree structure with $n$ nodes and $l$ leaves is $3n-2l-3$ bits.
\end{theorem}
\begin{proof}
The number of $\rightarrow$s and $\leftarrow$s used in the code is exactly equal to the total number of nodes minus the root. However, we use two bits for each $\rightarrow$, which shows non-leaf nodes. Therefore, the $\rightarrow$s and $\leftarrow$s use $2(n-1)-l$ bits in total. Additionally, for every node that has at least one child (except for the root), we will have a $\Rightarrow$ in the code. This would be equal to $n-1-l$ bits. Thus, we will have $3n-2l-3$ bits in total.
\end{proof}

\section{TreeExplorer}
\label{hybrid}

To see in which scenarios TD performs better than PC we can write
\begin{equation}
\label{inequa}
    3n-2l-3<n+2l-3 \Rightarrow l>n/2.
\end{equation}
Based on Eq. (\ref{inequa}), PC works better when the number of leaves is less than $n/2$, and TD works better otherwise. They exhibit the same performance when the tree has exactly $n/2$ leaves. Based on this, we propose the following coding technique for rooted tree structures.

\textbf{TreeExplorer:} Firstly, the number of leaves of the rooted tree structure is counted ($l$). If $l< n/2$, the structure is coded with PC. The code is then prefixed with a 0 to specify that it has been coded using PC. Otherwise, the structure is coded using TD, and the code is prefixed with a 1.

It can easily be shown that the codes created using TreeExplorer are uniquely decodable. This is because the first bit of the code uniquely determines the coding method, and we have already proven that both PC and TD are uniquely decodable. The following theorem shows the upper bound for the average code length of TreeExplorer.

\begin{theorem}
For any probability distribution on rooted tree structures with $n$ nodes, the average code length of TreeExplorer is less than $2n-2$.
\end{theorem}
\begin{proof}
Assume that the probability of the number of leaves ($l$) being less than $n/2$ is $q$. If $L$ is the length of the code produced using TreeExplorer, we can write
\begin{align*}
    \mathbb{E}[L] & = q\mathop{\mathbb{E}}_{l<n/2}[n+2l-2]+(1-q)\mathop{\mathbb{E}}_{l\geq n/2}[3n-2l-2]\\
    &< q(2n-2)+(1-q)(2n-2) \\
    & = 2n-2.
\end{align*}
\end{proof}

\section{Results and comparison}
\label{results}
In this section, we provide some results from simulating TreeExplorer. For all the simulations in this section, the average code length is calculated by uniformly creating a pool of sample rooted tree structures, coding them, and then averaging the code length.

\subsection{Comparison with entropy}

We compare the average code length of TreeExplorer with the entropy of the uniform source, which was calculated in section \ref{sources}. The result is plotted in Fig. \ref{entropy_uni}. It can be seen that the performance of TreeExplorer is very close to the entropy of the source, which is the optimal compression limit.

\begin{figure}
    \centering
     \begin{subfigure}[b]{\columnwidth}
         \centering
         \includegraphics[width=0.7\columnwidth]{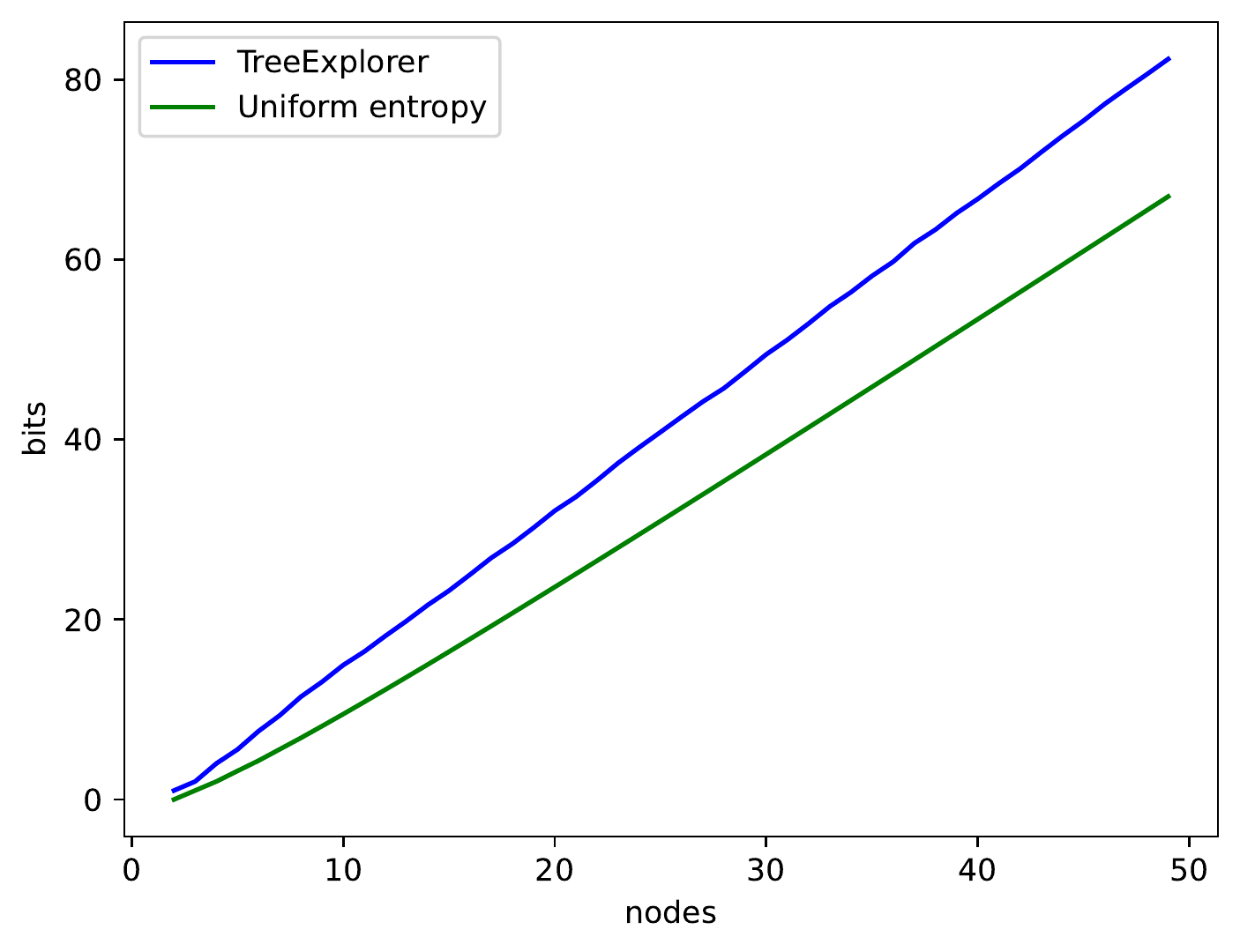}
         \caption{Average length}
         \label{length}
     \end{subfigure}
     \hfill
     \begin{subfigure}[b]{\columnwidth}
         \centering
         \includegraphics[width=0.7\columnwidth]{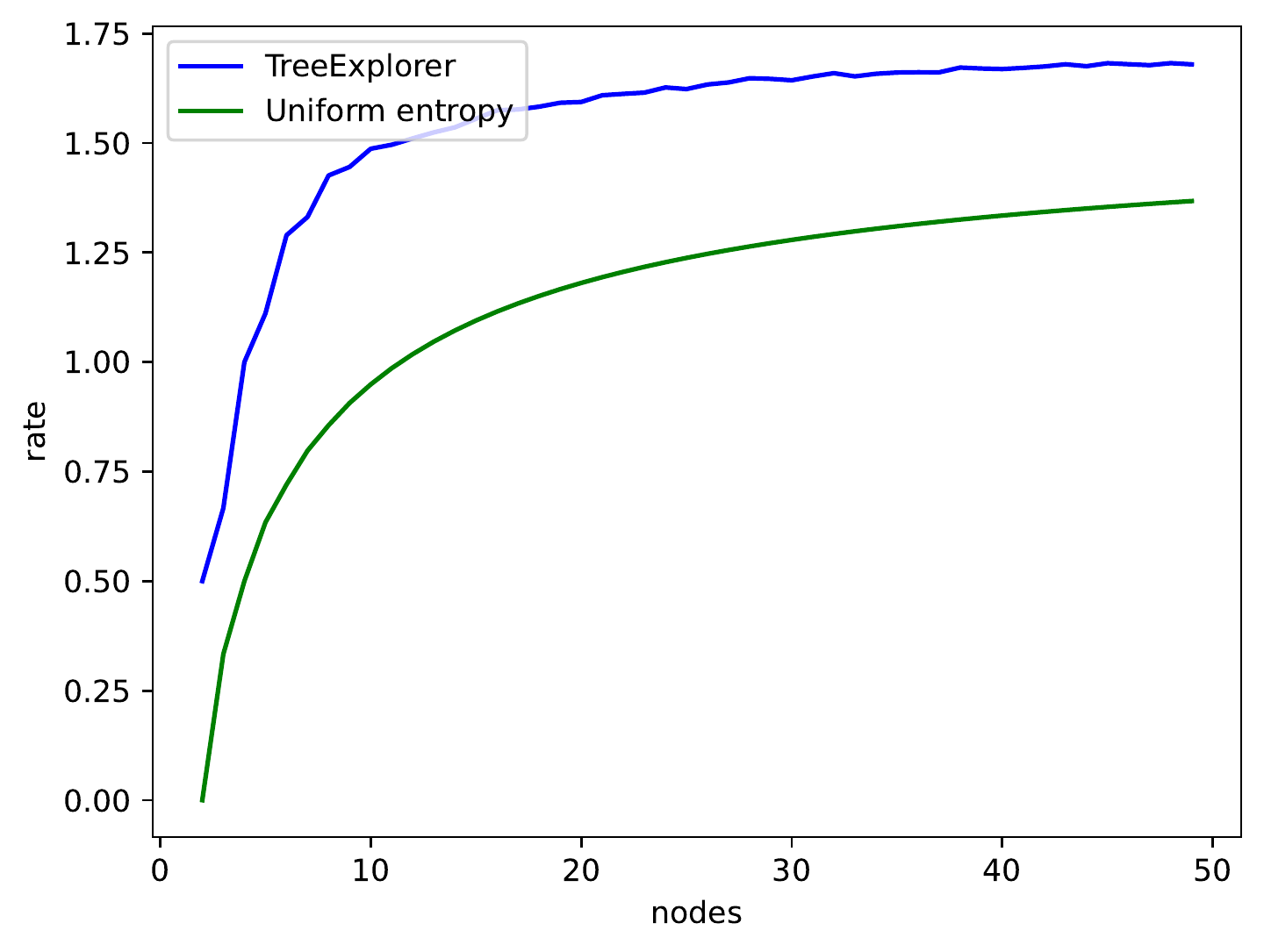}
         \caption{Average rate of change}
         \label{rate}
     \end{subfigure}
     \hfill
     \caption{Comparing the performance of TreeExplorer on a uniform unlabeled unordered rooted tree source with its entropy. Subfigure \ref{length} shows the average bit length, and subfigure \ref{rate} shows its average rate of change.}
    \label{entropy_uni}
\end{figure}

\subsection{Adjacency list}

The adjacency list representation is one of the most widely used methods for storing trees. This method uses $2n\lceil \log_2 n\rceil$ bits to represent a tree. Fig. \ref{adj list} compares the performance of adjacency list with TreeExplorer. Notice that for coding a labeled tree using TreeExplorer, an additional $n\lceil \log_2 n\rceil$ bits are needed to list the node labels in the order of their appearance.

\begin{figure}
    \centering
    \includegraphics[width = 0.7\columnwidth]{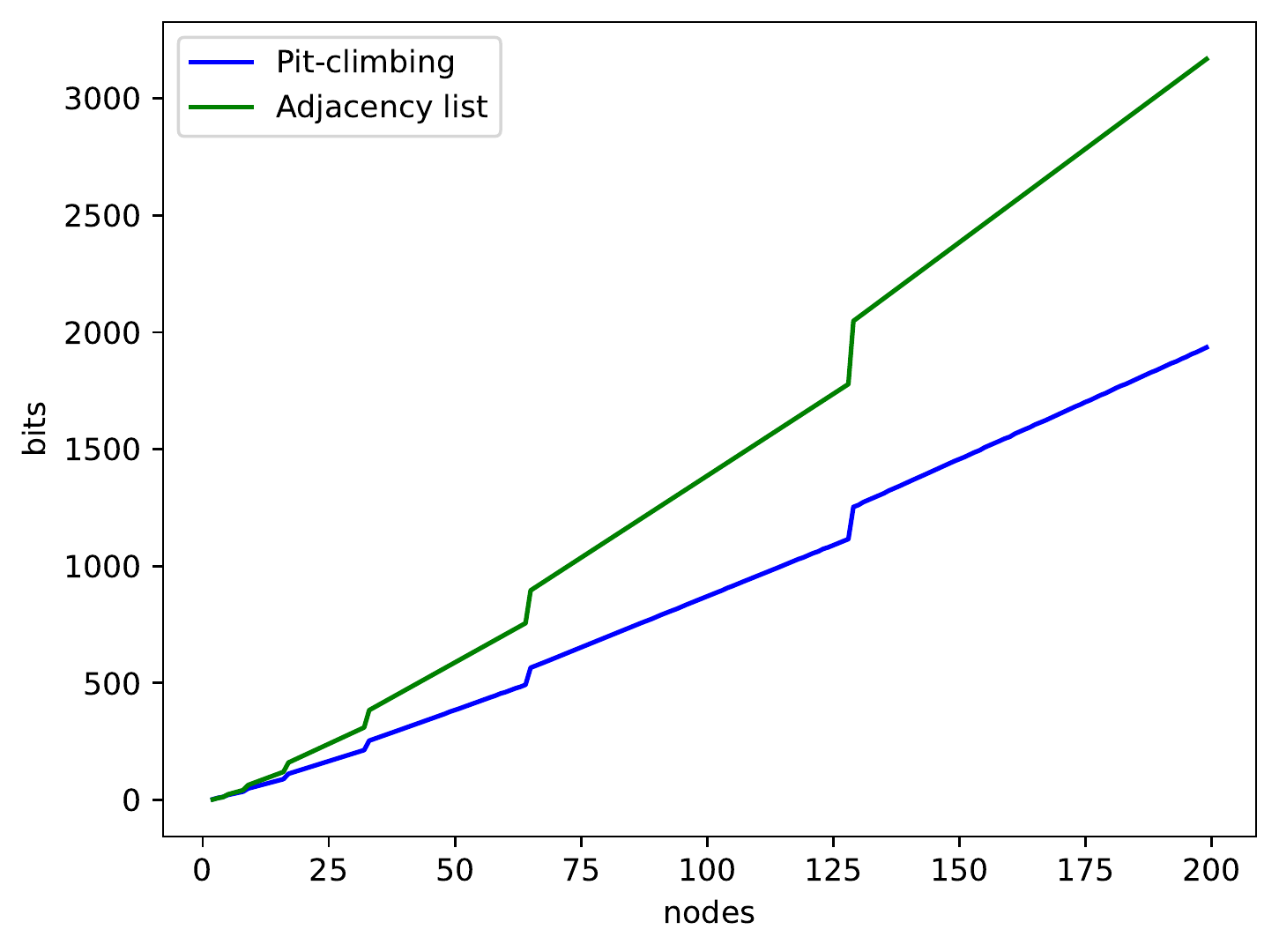}
    \caption{Comparing the performance of TreeExplorer with adjacency list}
    \label{adj list}
\end{figure}

\subsection{The Newick format}

The Newick format has been the standard for representing phylogenetic trees since its introduction back in 1986 \cite{cardona2008extended}. In this method, trees are represented using parentheses and commas. This format starts from the root of the tree, and lists the children and subtrees of the root in a nested manner. We will not go into further details on how this method works. However, because of the similarities between this format and TreeExplorer, we compare the performance of these two methods.

Our calculations show an additional $n+1$ bits and $3n-4l+1$ bits in the Newick format compared to when TreeExplorer uses TD and PC, respectively. Fig. \ref{Newick} compares the average codeword lengths of TreeExplorer and the Newick format for rooted trees of up to 50 nodes. It can be seen that the figure also confirms that TreeExplorer provides us with a shorter codeword length.

\begin{figure}[h!]
    \centering
    \includegraphics[width = 0.7\columnwidth]{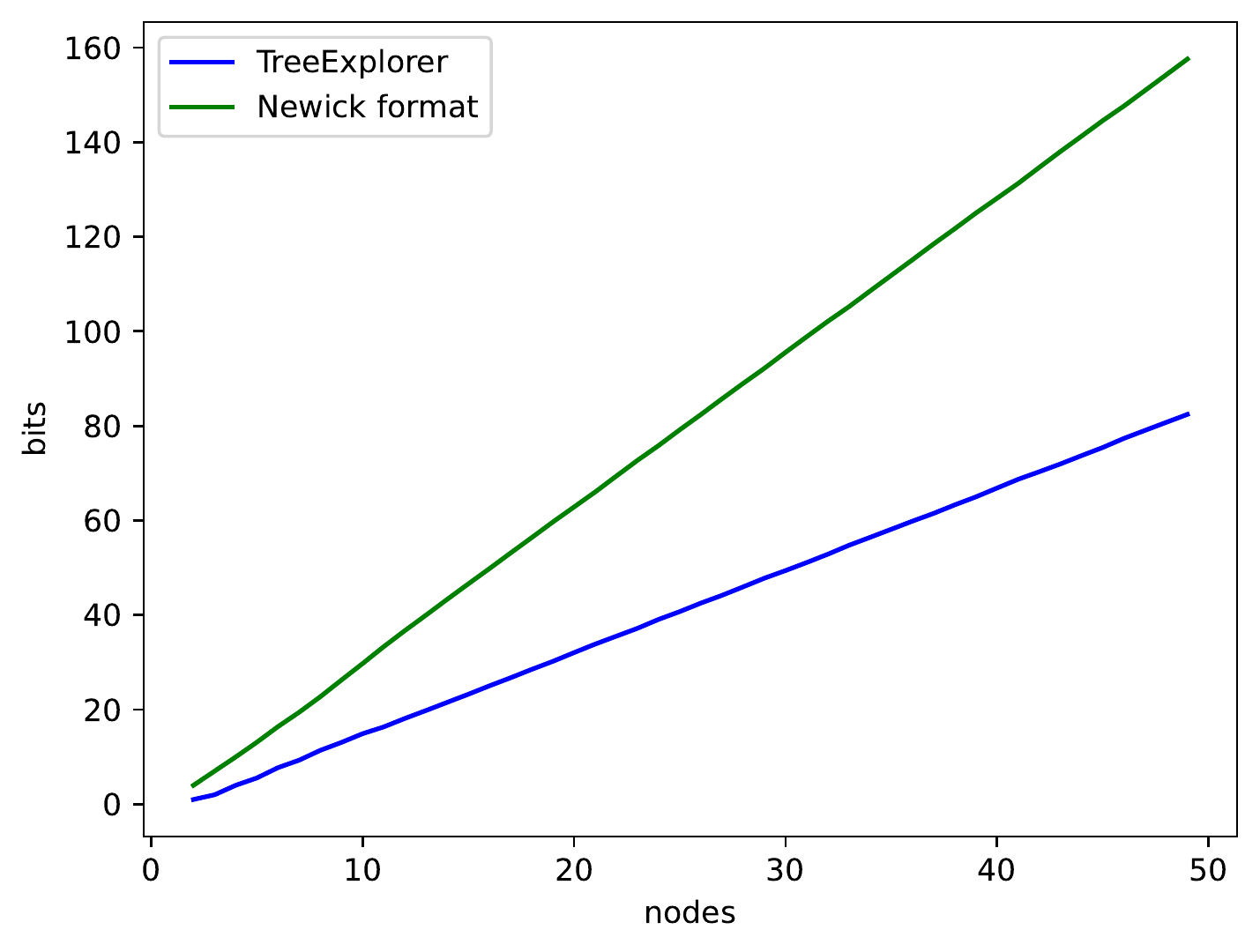}
    \caption{Comparing the performance of TreeExplorer with the Newick format}
    \label{Newick}
\end{figure}

\section{Application to network routing}

As stated earlier, one of the main motivations behind introducing a new coding algorithm for rooted trees is its application to routing protocols in ad hoc and wireless networks. In this section, we will explain how TreeExplorer can be applied to the existing routing protocols and improve their performance.

There are two sides to the information that we want to store or communicate about a routing table: the structure of the underlying rooted tree, and the labels of the nodes in the tree. Our proposed method of coding a routing table is to first code the underlying unlabeled structure using TreeExplorer, and then follow that with the labels of the nodes in the tree, in the order in which they are traversed in TreeExplorer. Fig. \ref{message} shows how a message packet will look like if we code the routing table using the proposed method.

\begin{figure}[h!]
    \centering
    \includegraphics[width = \columnwidth]{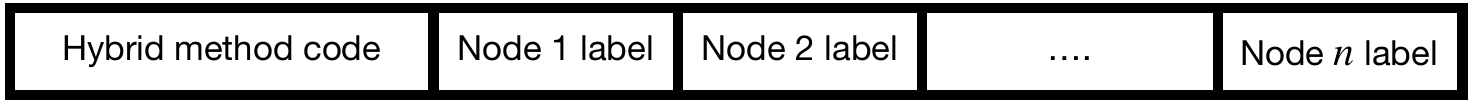}
    \caption{A message packet containing a routing table coded using TreeExplorer}
    \label{message}
\end{figure}

Path-vector routing protocol uses protocol messages to communicate the shortest paths between nodes. Each of these messages contains all the hops in the shortest path from a destination node to a source node. This will result in a redundancy when describing different paths which have mutual edges. For example, if node B is included in the shortest path from node A to node C, then the shortest path from A to B will be repeated in at least two protocol messages, one for the shortest path from A to B, and the other one for the shortest path from A to C. This redundancy will not exist if we communicate the routing table using a single rooted tree.

Another important note is that many of the changes to networks happen either to the underlying topology or the node labels. Having the description for these two separately gives us the freedom to update them individually. For example, if there is a change in the structure but not the node labels, the structure can easily and efficiently be updated using TreeExplorer.
\balance

\section{Conclusion}
In this paper, TreeExplorer was introduced as a novel coding technique for unlabeled rooted trees. In essence, TreeExplorer is a hybrid of two other methods, pit-climbing and tunnel-digging. TreeExplorer chooses one of these two methods based on the number of leaves of the structure that it is supposed to code. Even though TreeExplorer was mainly introduced for unordered trees, it can easily be extended to ordered trees as well. For an ordered tree, the order in which the tree will be explored will be the same as the order of the branches. Simulation results exhibited a near-optimal performance in terms of the code length, and it was shown that TreeExplorer has a better performance than simliar existing techniques. Finally, it was discussed how TreeExplorer can directly impact and enhance ad hoc and wireless routing protocols. This can be beneficial as significant communication resources are currently being used for these protocols, and the rapidly growing number of devices in the network is making routing protocols much more challenging than before.

\section*{Acknowledgment}

This research was funded in whole or in part by EPSRC grant number EP/T02612X/1. For the purpose of Open Access, the author has applied a CC BY public copyright licence to any Author Accepted Manuscript (AAM) version arising from this submission. We also thank Moogsoft inc. for their support in this research project.

\bibliographystyle{ieeetr}
\bibliography{bib}

\begin{thebibliography}{10}

\bibitem{felsenstein2004inferring}
J.~Felsenstein, {\em Inferring phylogenies}, vol.~2.
\newblock Sinauer associates Sunderland, MA, 2004.

\bibitem{dykes2011xml}
L.~Dykes and E.~Tittel, {\em XML for Dummies}.
\newblock John Wiley \& Sons, 2011.

\bibitem{cormen2022introduction}
T.~H. Cormen, C.~E. Leiserson, R.~L. Rivest, and C.~Stein, {\em Introduction to
  algorithms}.
\newblock MIT press, 2022.

\bibitem{medhi2017network}
D.~Medhi and K.~Ramasamy, {\em Network routing: algorithms, protocols, and
  architectures}.
\newblock Morgan kaufmann, 2017.

\bibitem{sobrinho2003network}
J.~L. Sobrinho, ``Network routing with path vector protocols: Theory and
  applications,'' in {\em Proceedings of the 2003 conference on Applications,
  technologies, architectures, and protocols for computer communications},
  pp.~49--60, 2003.

\bibitem{chau2008inter}
C.-K. Chau, J.~Crowcroft, K.-W. Lee, and S.~H. Wong, ``Inter-domain routing for
  mobile ad hoc networks,'' in {\em Proceedings of the 3rd international
  workshop on Mobility in the evolving internet architecture}, pp.~61--66,
  2008.

\bibitem{rangarajan2004using}
H.~Rangarajan and J.~Garcia-Luna-Aceves, ``Using labeled paths for loop-free
  on-demand routing in ad hoc networks,'' in {\em Proceedings of the 5th ACM
  international symposium on Mobile ad hoc networking and computing},
  pp.~43--54, 2004.

\bibitem{zhang2007ad}
C.~Zhang, M.~Zhou, and M.~Yu, ``Ad hoc network routing and security: A
  review,'' {\em International Journal of Communication Systems}, vol.~20,
  no.~8, pp.~909--925, 2007.

\bibitem{ioannidis2004scalable}
I.~Ioannidis and B.~Carbunar, ``Scalable routing in hybrid cellular and ad-hoc
  networks,'' in {\em 2004 IEEE International Conference on Mobile Ad-hoc and
  Sensor Systems (IEEE Cat. No. 04EX975)}, pp.~522--524, IEEE, 2004.

\bibitem{ioannidis2005high}
I.~Ioannidis, B.~Carbunar, and C.~Nita-Rotaru, ``High throughput routing in
  hybrid cellular and ad-hoc networks,'' in {\em Sixth IEEE International
  Symposium on a World of Wireless Mobile and Multimedia Networks},
  pp.~171--176, IEEE, 2005.

\bibitem{farzaneh2021kolmogorov}
A.~Farzaneh, J.~P. Coon, and M.-A. Badiu, ``Kolmogorov basic graphs and their
  application in network complexity analysis,'' {\em Entropy}, vol.~23, no.~12,
  p.~1604, 2021.

\bibitem{farzaneh2022information}
A.~Farzaneh and J.~P. Coon, ``An information theory approach to network
  evolution models,'' {\em Journal of Complex Networks}, vol.~10, no.~3,
  p.~cnac020, 2022.

\bibitem{nakamura20205g}
T.~Nakamura, ``5g evolution and 6g,'' in {\em 2020 IEEE symposium on VLSI
  technology}, pp.~1--5, IEEE, 2020.

\bibitem{goldsmith2005wireless}
A.~Goldsmith, {\em Wireless communications}.
\newblock Cambridge university press, 2005.

\bibitem{shannon1948mathematical}
C.~E. Shannon, ``A mathematical theory of communication,'' {\em The Bell system
  technical journal}, vol.~27, no.~3, pp.~379--423, 1948.

\bibitem{oeis}
{OEIS Foundation Inc.}, ``The {O}n-{L}ine {E}ncyclopedia of {I}nteger
  {S}equences,'' 2022.
\newblock Published electronically at \url{http://oeis.org}.

\bibitem{polya2012combinatorial}
G.~Polya and R.~C. Read, {\em Combinatorial enumeration of groups, graphs, and
  chemical compounds}.
\newblock Springer Science \& Business Media, 2012.

\bibitem{koshy2008catalan}
T.~Koshy, {\em Catalan numbers with applications}.
\newblock Oxford University Press, 2008.

\bibitem{cayley_2009}
A.~Cayley, {\em A theorem on trees}, vol.~13 of {\em Cambridge Library
  Collection - Mathematics}, p.~26–28.
\newblock Cambridge University Press, 2009.

\bibitem{cardona2008extended}
G.~Cardona, F.~Rossell{\'o}, and G.~Valiente, ``Extended newick: it is time for
  a standard representation of phylogenetic networks,'' {\em BMC
  bioinformatics}, vol.~9, no.~1, pp.~1--8, 2008.

\end{thebibliography}


\begin{thebibliography}{00}
\bibitem{b1} G. Eason, B. Noble, and I. N. Sneddon, ``On certain integrals of Lipschitz-Hankel type involving products of Bessel functions,'' Phil. Trans. Roy. Soc. London, vol. A247, pp. 529--551, April 1955.
\bibitem{b2} J. Clerk Maxwell, A Treatise on Electricity and Magnetism, 3rd ed., vol. 2. Oxford: Clarendon, 1892, pp.68--73.
\bibitem{b3} I. S. Jacobs and C. P. Bean, ``Fine particles, thin films and exchange anisotropy,'' in Magnetism, vol. III, G. T. Rado and H. Suhl, Eds. New York: Academic, 1963, pp. 271--350.
\bibitem{b4} K. Elissa, ``Title of paper if known,'' unpublished.
\bibitem{b5} R. Nicole, ``Title of paper with only first word capitalized,'' J. Name Stand. Abbrev., in press.
\bibitem{b6} Y. Yorozu, M. Hirano, K. Oka, and Y. Tagawa, ``Electron spectroscopy studies on magneto-optical media and plastic substrate interface,'' IEEE Transl. J. Magn. Japan, vol. 2, pp. 740--741, August 1987 [Digests 9th Annual Conf. Magnetics Japan, p. 301, 1982].
\bibitem{b7} M. Young, The Technical Writer's Handbook. Mill Valley, CA: University Science, 1989.
\end{thebibliography}

\end{document}